\documentclass[conference]{IEEEtran}
\usepackage{amsmath}
\usepackage{amssymb}
\usepackage{amsthm}
\usepackage[T1]{fontenc}
\usepackage{tikz}

\newcommand{\betaqspace}{\hspace{.05em}}
\newcommand{\expectation}[1]{\operatorname{E}[#1]}
\newcommand{\numberoflabels}{\mathsf{M}}
\newcommand{\qxcommaqy}{Q_X\hspace{-.1em},\hspace{.1em}Q_Y}
\newcommand{\ratex}{\mathsf{R_X}}
\newcommand{\ratey}{\mathsf{R_Y}}
\newcommand{\rhotilde}{\tilde{\rho}}
\newcommand{\set}[1]{\mathcal{#1}}
\newtheorem{lemma}{Lemma}
\newtheorem{remark}{Remark}
\newtheorem{theorem}{Theorem}

\usetikzlibrary{arrows.meta}

\begin{document}
\title{Distributed Task Encoding}
\author{\IEEEauthorblockN{Annina Bracher}
\IEEEauthorblockA{Group Risk Management\\
Swiss Re, 8022 Zurich, Switzerland\\
Email: annina.bracher@gmail.com}
\and
\IEEEauthorblockN{Amos Lapidoth and Christoph Pfister}
\IEEEauthorblockA{Signal and Information Processing Laboratory\\
ETH Zurich, 8092 Zurich, Switzerland\\
Email: \{lapidoth,pfister\}@isi.ee.ethz.ch}}

\maketitle

\begin{abstract}
The rate region of the task-encoding problem for two correlated sources is characterized using a novel parametric family of dependence measures.
The converse uses a new expression for the $\rho$-th moment of the list size, which is derived using the relative $\alpha$-entropy.
\end{abstract}

\section{Introduction}
We extend the task-encoding problem introduced by Bunte and Lapidoth \cite{BasicTaskEncoding} to the distributed setting depicted in Figure~\ref{fig_intro_settting}.
A source generates a sequence of pairs $\{(X_i,Y_i)\}_{i=1}^{n}$ over the finite alphabet $\set{X} \times \set{Y}$.
Using the functions
\begin{IEEEeqnarray}{rl}
f_n\colon \set{X}^{n} \to \{&1,\ldots,\lfloor 2^{n \ratex} \rfloor\}, \\
g_n\colon \set{Y}^{n} \to \{&1,\ldots,\lfloor 2^{n \ratey} \rfloor\}, \IEEEeqnarraynumspace
\end{IEEEeqnarray}
the sequence $X^{n}$ is described by one of $\lfloor 2^{n \ratex} \rfloor$ labels and the sequence $Y^{n}$ by one of $\lfloor 2^{n \ratey} \rfloor$ labels.
The decoder outputs the list of all pairs of sequences that could have produced the given pair of labels.
The size of this list is
\begin{IEEEeqnarray}{rl}
L(X^{n},Y^{n}) \triangleq |\{&(x',y') \in \set{X}^{n} \times \set{Y}^{n}: \nonumber \\*
& f_n(x') = f_n(X^{n}) \> \land \> g_n(y') = g_n(Y^{n}) \}|. \IEEEeqnarraynumspace \label{eq_intro_lxy_def}
\end{IEEEeqnarray}
For a fixed $\rho > 0$, a rate pair $(\ratex,\ratey)$ is called achievable if there exists a sequence of task encoders $\{(f_n,g_n)\}_{n=1}^{\infty}$ such that the $\rho$-th moment of the list size tends to one as $n$ tends to infinity, i.e., if
\begin{IEEEeqnarray}{C}
\lim_{n \to \infty} \expectation{L(X^{n},Y^{n})^{\rho}} = 1. \IEEEeqnarraynumspace \label{eq_intro_lim_exp_one}
\end{IEEEeqnarray}
Our main contribution is Theorem~\ref{thm_rate_region}, which states that rate pairs $(\ratex,\ratey)$ in the interior of the following region are achievable, while those outside the region are not:
\begin{IEEEeqnarray}{rCl}
\ratex &\ge& \limsup_{n \to \infty} \frac{H_{\rhotilde}(X^{n})}{n}, \label{eq_intro_dist_region_gen_a} \\
\ratey &\ge& \limsup_{n \to \infty} \frac{H_{\rhotilde}(Y^{n})}{n}, \label{eq_intro_dist_region_gen_b} \\
\ratex + \ratey &\ge& \limsup_{n \to \infty} \frac{H_{\rhotilde}(X^{n}, Y^{n}) + K_{\rhotilde}(X^{n}; Y^{n})}{n}, \IEEEeqnarraynumspace \label{eq_intro_dist_region_gen_c}
\end{IEEEeqnarray}
where $H_{\rhotilde}$ denotes the R\'enyi entropy, $K_{\rhotilde}$ is a dependence measure defined in Section~\ref{sec_definitions}, and throughout the paper
\begin{IEEEeqnarray}{C}
\rhotilde \triangleq \frac{1}{1+\rho}. \IEEEeqnarraynumspace \label{eq_intro_rhotilde}
\end{IEEEeqnarray}
In the IID case, (\ref{eq_intro_dist_region_gen_a})--(\ref{eq_intro_dist_region_gen_c}) reduce to
\begin{IEEEeqnarray}{rCl}
\ratex &\ge& H_{\rhotilde}(P_X), \label{eq_intro_iid_region_a} \\
\ratey &\ge& H_{\rhotilde}(P_Y), \label{eq_intro_iid_region_b} \\
\ratex + \ratey &\ge& H_{\rhotilde}(P_{XY}) + K_{\rhotilde}(X; Y). \IEEEeqnarraynumspace \label{eq_intro_iid_region_c}
\end{IEEEeqnarray}

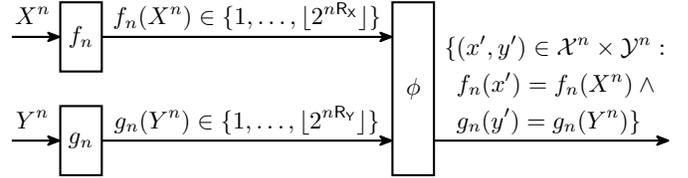
\begin{figure}[!t]
\centering
\begin{tikzpicture}[thick,>={Stealth[round]},scale=0.92,every node/.style={transform shape}]
\draw[->] (-0.85cm, 0.75cm) -- (-0.15cm, 0.75cm);
\draw[->] (-0.85cm,-0.75cm) -- (-0.15cm,-0.75cm);
\node[inner sep=0.4mm,anchor=south] at (-0.55cm, 0.75cm) {$\begin{aligned}[t]\!X^{n}\end{aligned}$};
\node[inner sep=0.4mm,anchor=south] at (-0.55cm,-0.75cm) {$\begin{aligned}[t]\!Y^{n}\end{aligned}$};
\node[draw,inner sep=0,minimum height=1cm,minimum width=0.6cm] at (0.15cm, 0.75cm) {$f_n$};
\node[draw,inner sep=0,minimum height=1cm,minimum width=0.6cm] at (0.15cm,-0.75cm) {$g_n$};
\draw[->] (0.45cm, 0.75cm) -- (4.65cm, 0.75cm);
\draw[->] (0.45cm,-0.75cm) -- (4.65cm,-0.75cm);
\node[inner sep=0.4mm,anchor=south] at (2.55cm, 0.75cm) {$\begin{aligned}[t]\!f_n(X^{n}) \in \{1,\ldots,\lfloor 2^{n \ratex} \rfloor\}\end{aligned}$};
\node[inner sep=0.4mm,anchor=south] at (2.55cm,-0.75cm) {$\begin{aligned}[t]\!g_n(Y^{n}) \in \{1,\ldots,\lfloor 2^{n \ratey} \rfloor\}\end{aligned}$};
\node[draw,inner sep=0,minimum height=2.5cm,minimum width=0.6cm] at (4.95cm,0cm) {$\phi$};
\draw[->] (5.25cm,-0.75cm) -- (8.65cm,-0.75cm);
\node[inner sep=0.4mm,anchor=south] at (7.0cm, -0.75cm) {$\begin{aligned}[t]\!\{&(x',y') \in \set{X}^{n} \times \set{Y}^{n}: \\
&f_n(x') = f_n(X^{n}) \> \land \\
&g_n(y') = g_n(Y^{n})\}\end{aligned}$};
\end{tikzpicture}
\caption{Distributed task encoding with encoders $f_n$ and $g_n$ and a decoder $\phi$.}
\label{fig_intro_settting}
\end{figure}

Compared to Slepian--Wolf coding \cite{CoverAndThomas}, we notice two major differences.
First, the constraints (\ref{eq_intro_iid_region_a}) and (\ref{eq_intro_iid_region_b}) only depend on the marginal PMFs $P_X$ and $P_Y$, so the information that $X$ reveals about $Y$ and vice-versa has no influence on these constraints.
Second, the constraint on the sum rate includes a term $K_{\rhotilde}$, which is not present in the single-source setting \cite[Theorem~I.2]{BasicTaskEncoding}. (The term $K_{\rhotilde}$ is always nonnegative and zero if and only if $X$ and $Y$ are independent \cite[Theorem~2]{TwoMeasuresOfDependence}.)

Task encoding is related to the Massey--Arikan guessing experiment \cite{MasseyGuessing,ArikanGuessing}, where the decoder repeatedly guesses until correct.
While the guessing problem and the task-encoding problem lead to the same asymptotics in the single-source setting \cite{GuessingAndTaskEncoding}, this is no longer the case in the distributed setting \cite{DistributedGuessingAndTaskEncoding}:
except if $X$ and $Y$ are independent, the guessing region from Section~\ref{sec_dist_guessing} is strictly larger than the task-encoding region (\ref{eq_intro_iid_region_a})--(\ref{eq_intro_iid_region_c}).

Another contribution concerns the $\rho$-th moment of the list size in the single-source setting.
Let $\set{X}$ be a finite set of tasks from which a task $X$ is drawn at random according to the PMF $P$ and then mapped to one of $\numberoflabels$ labels by a task encoder $f\colon \set{X} \to \{1,\ldots,\numberoflabels\}$.
Given a task $x$, we denote by
\begin{IEEEeqnarray}{C}
L(x) \triangleq |\{x' \in \set{X}: f(x') = f(x)\}| \IEEEeqnarraynumspace \label{eq_intro_listsize}
\end{IEEEeqnarray}
the size of the list, i.e., the number of tasks that have the same label as $x$.
In Lemma~\ref{lma_single_cost}, we show that the $\rho$-th moment of the list size can be expressed as
\begin{IEEEeqnarray}{C}
\expectation{L(X)^{\rho}} = 2^{\rho [H_{\rhotilde}(P) + \Delta_{\rhotilde}(P||Q) - \log \numberoflabels']}, \IEEEeqnarraynumspace \label{eq_intro_cost_decomp}
\end{IEEEeqnarray}
where $H_{\rhotilde}$ and $\Delta_{\rhotilde}$ are the R\'enyi entropy of order $\rhotilde$ and the relative $\rhotilde$-entropy, respectively, which will be defined in Section~\ref{sec_definitions};
$Q$ is an auxiliary PMF that depends only on the task encoder $f$; and
$\numberoflabels'$ equals the number of used labels.
The analogy between (\ref{eq_intro_cost_decomp}) and a similar expression in classical fixed-to-variable length source coding is discussed at the end of Section~\ref{sec_el_rho_hdm}.

The remainder of this paper is organized as follows.
In Section~\ref{sec_definitions}, we define R\'enyi's information measures and review some of their properties.
In Section~\ref{sec_el_rho_hdm}, we prove (\ref{eq_intro_cost_decomp}) and draw the analogy between (\ref{eq_intro_cost_decomp}) and a similar expression in classical fixed-to-variable length source coding.
In Section~\ref{sec_distributed_te}, we show that (\ref{eq_intro_dist_region_gen_a})--(\ref{eq_intro_dist_region_gen_c}) characterize the region of achievable rate pairs for distributed task encoding.
In Section~\ref{sec_discussion}, we compare (\ref{eq_intro_listsize}) with the related setting where the decoder's list only contains tasks with positive posterior probability.
In Section~\ref{sec_dist_guessing}, we discuss the guessing problem for two correlated sources.

\section{R\'enyi's Information Measures}
\label{sec_definitions}

All logarithms in this paper are to base two.
The R\'enyi entropy of order $\alpha$ was introduced by R\'enyi \cite{RenyiEntropyDivergence} and is defined for $\alpha > 0$ and $\alpha \ne 1$ as
\begin{IEEEeqnarray}{C}
H_\alpha(P) \triangleq \frac{1}{1-\alpha} \log \sum_{x} P(x)^{\alpha}, \IEEEeqnarraynumspace \label{eq_def_renyi_entropy}
\end{IEEEeqnarray}
where $P$ is a PMF.
It is a generalization of Shannon entropy because $\lim_{\alpha \to 1} H_\alpha(P) = H(P)$.
If the PMF of $X$ is $P_X$, we also use $H_\alpha(X)$ to denote $H_\alpha(P_X)$.

The R\'enyi divergence of order $\alpha$ was also introduced by R\'enyi \cite{RenyiEntropyDivergence} and is defined for $\alpha > 0$ and $\alpha \ne 1$ as
\begin{IEEEeqnarray}{C}
D_\alpha(P||Q) \triangleq \frac{1}{\alpha-1} \log \sum_{x} P(x)^{\alpha} Q(x)^{1-\alpha}, \IEEEeqnarraynumspace \label{eq_def_renyi_div}
\end{IEEEeqnarray}
where $P$ and $Q$ are PMFs and where we use the convention that for $\alpha > 1$, we read $P(x)^{\alpha} Q(x)^{1-\alpha}$ as \raisebox{0pt}[0pt][0pt]{$\frac{P(x)^{\alpha}}{Q(x)^{\alpha-1}}$} and say that $\frac{0}{0} = 0$ and $\frac{p}{0} = \infty$ for $p > 0$.
It is a generalization of Kullback--Leibler divergence because $\lim_{\alpha \to 1} D_\alpha(P||Q)$ is equal to $D(P||Q)$.

The relative $\alpha$-entropy was defined by Sundaresan \cite{SundaresanISIT, SundaresanGuessing} for $\alpha > 0$ and $\alpha \ne 1$ as
\begin{IEEEeqnarray}{rCl}
\Delta_\alpha(P||Q) &\triangleq& \frac{\alpha}{1-\alpha} \log \sum_{x} P(x) Q(x)^{\alpha-1} \nonumber \\*
&& +\>\log \sum_{x} Q(x)^{\alpha} - \frac{1}{1-\alpha} \log \sum_{x} P(x)^{\alpha}, \IEEEeqnarraynumspace \label{eq_def_rel_a_div}
\end{IEEEeqnarray}
where $P$ and $Q$ are PMFs and where we use the convention that for $\alpha < 1$, we read $P(x) Q(x)^{\alpha-1}$ as \raisebox{0pt}[0pt][0pt]{$\frac{P(x)}{Q(x)^{1-\alpha}}$} and say that $\frac{0}{0} = 0$ and $\frac{p}{0} = \infty$ for $p > 0$.
It is also a generalization of Kullback--Leibler divergence because $\lim_{\alpha \to 1} \Delta_\alpha(P||Q)$ is equal to $D(P||Q)$.

Relative $\alpha$-entropy and R\'enyi divergence are related as follows \cite[Lemma~1]{TwoMeasuresOfDependence}:
\begin{IEEEeqnarray}{C}
\Delta_\alpha(P||Q) = D_\frac{1}{\alpha}(\hat{P}||\hat{Q}), \IEEEeqnarraynumspace \label{eq_def_delta_to_d}
\end{IEEEeqnarray}
where the transformed PMFs $\hat{P}$ and $\hat{Q}$ are given by
\begin{IEEEeqnarray}{rCl}
\hat{P}(x) &\triangleq& \frac{P(x)^{\alpha}}{\sum_{x'} P(x')^{\alpha}}, \IEEEeqnarraynumspace \label{eq_def_p_hat} \\
\hat{Q}(x) &\triangleq& \frac{Q(x)^{\alpha}}{\sum_{x'} Q(x')^{\alpha}}. \label{eq_def_q_hat}
\end{IEEEeqnarray}
For a fixed $\alpha > 0$, this transformation is bijective on the set of all PMFs because for all $x \in \set{X}$,
\begin{IEEEeqnarray}{C}
Q(x) = \frac{\hat{Q}(x)^{1/\alpha}}{\sum_{x'} \hat{Q}(x')^{1/\alpha}}. \IEEEeqnarraynumspace \label{eq_def_q_hat_inverse}
\end{IEEEeqnarray}

The measure of dependence $K_\alpha(X;Y)$ was introduced in \cite{TwoMeasuresOfDependence} and is defined as
\begin{IEEEeqnarray}{C}
K_\alpha(X;Y) \triangleq \min_{\qxcommaqy} \Delta_\alpha(P_{XY}||Q_X Q_Y), \IEEEeqnarraynumspace \label{eq_def_k_alpha}
\end{IEEEeqnarray}
where $P_{XY}$ is the joint PMF of $X$ and $Y$ and the minimization is over all PMFs $Q_X$ and $Q_Y$.
It is a generalization of the mutual information because $\lim_{\alpha \to 1} K_\alpha(X;Y) = I(X;Y)$.

Recalling (\ref{eq_intro_rhotilde}) and (\ref{eq_def_renyi_entropy}), we obtain from (\ref{eq_def_rel_a_div})
\begin{IEEEeqnarray}{rCl}
\Delta_{\rhotilde}(P||Q) &=& \frac{1}{\rho} \log \sum_{x} P(x) Q(x)^{-\rho \rhotilde} \nonumber \\*
&& +\>\log \sum_{x} Q(x)^{\rhotilde} - H_{\rhotilde}(P). \IEEEeqnarraynumspace \label{eq_def_delta_rhotilde}
\end{IEEEeqnarray}

\section{Moments of the List Size}
\label{sec_el_rho_hdm}

\begin{lemma}
\label{lma_single_cost}
Let $P$ be a PMF on the finite set $\set{X}$, let $\numberoflabels$ be a positive integer, let $f$ be a function from $\set{X}$ to $\{1,\ldots,\numberoflabels\}$, let~$L$ be defined as in (\ref{eq_intro_listsize}), and let $\rho > 0$.
Define the PMF $Q$ as
\begin{IEEEeqnarray}{C}
Q(x) \triangleq \frac{L(x)^{-(1+\rho)}}{\sum_{x'} L(x')^{-(1+\rho)}}. \IEEEeqnarraynumspace \label{eq_lma_single_qdef}
\end{IEEEeqnarray}
If $X$ is distributed according to $P$, then
\begin{IEEEeqnarray}{C}
\expectation{L(X)^{\rho}} = 2^{\rho [H_{\rhotilde}(P) + \Delta_{\rhotilde}(P||Q) - \log \numberoflabels']}, \IEEEeqnarraynumspace \label{eq_lma_single_explr}
\end{IEEEeqnarray}
where $\numberoflabels'$ denotes the number of labels that are actually used (as opposed to allowed), i.e.,
\begin{IEEEeqnarray}{C}
\numberoflabels' \triangleq |\{f(x): x \in \set{X} \}| \le \numberoflabels. \IEEEeqnarraynumspace
\end{IEEEeqnarray}
\end{lemma}

\begin{proof}
Since $L(x) \ge 1$ for all $x \in \set{X}$, $Q$ is well-defined and indeed a PMF.
Rearranging (\ref{eq_lma_single_qdef}), we get
\begin{IEEEeqnarray}{C}
L(x) = \beta \betaqspace Q(x)^{-\rhotilde} \IEEEeqnarraynumspace \label{eq_lma_single_cost_lx}
\end{IEEEeqnarray}
for some positive $\beta$.
Let $\set{M}' \triangleq \{f(x): x \in \set{X} \}$ be the set of labels that are used, and observe that
\begin{IEEEeqnarray}{rCl}
\numberoflabels' &=& \sum_{m \in \set{M}'} 1 \label{eq_lma_single_cost_a} \\
&=& \sum_{m \in \set{M}'} \sum_{x: f(x)=m} L(x)^{-1} \label{eq_lma_single_cost_b} \\
&=& \sum_{m \in \set{M}'} \sum_{x: f(x)=m} \beta^{-1} Q(x)^{\rhotilde} \IEEEeqnarraynumspace \label{eq_lma_single_cost_c}\\
&=& \beta^{-1} \sum_{x} Q(x)^{\rhotilde}, \label{eq_lma_single_cost_d}
\end{IEEEeqnarray}
where (\ref{eq_lma_single_cost_a}) holds because $\numberoflabels' = |\set{M}'|$;
(\ref{eq_lma_single_cost_b}) holds because for all $x \in \set{X}$ with $f(x)=m$, $L(x) = |\{x' \in \set{X}: f(x') = m\}|$;
(\ref{eq_lma_single_cost_c}) follows from (\ref{eq_lma_single_cost_lx}); and
(\ref{eq_lma_single_cost_d}) holds because each $x$ appears exactly once on the RHS of (\ref{eq_lma_single_cost_c}).
Consequently, $\beta$ can be expressed as
\begin{IEEEeqnarray}{C}
\beta = \frac{1}{\numberoflabels'} \sum_{x} Q(x)^{\rhotilde}, \IEEEeqnarraynumspace \label{eq_lma_single_cost_beta}
\end{IEEEeqnarray}
and
\begin{IEEEeqnarray}{rCl}
\IEEEeqnarraymulticol{3}{l}{\expectation{L(X)^{\rho}}} \nonumber \\* \quad
&=& \sum_{x} P(x) L(x)^{\rho} \label{eq_lma_single_cost_g} \\
&=& \beta^{\rho} \sum_{x} P(x) Q(x)^{-\rho \rhotilde} \label{eq_lma_single_cost_h} \\
&=& 2^{\rho [-\log \numberoflabels' + \log \sum_{x} Q(x)^{\rhotilde} + \frac{1}{\rho} \log \sum_{x} P(x) Q(x)^{-\rho \rhotilde}]} \IEEEeqnarraynumspace \label{eq_lma_single_cost_i} \\
&=& 2^{\rho [H_{\rhotilde}(P) + \Delta_{\rhotilde}(P||Q) - \log \numberoflabels']}, \label{eq_lma_single_cost_j}
\end{IEEEeqnarray}
where (\ref{eq_lma_single_cost_h}) follows from (\ref{eq_lma_single_cost_lx});
(\ref{eq_lma_single_cost_i}) follows from (\ref{eq_lma_single_cost_beta}); and
(\ref{eq_lma_single_cost_j}) follows from (\ref{eq_def_delta_rhotilde}).
\end{proof}

\begin{remark}
For every binary fixed-to-variable length source code, we have \cite[(5.25)]{CoverAndThomas}
\begin{IEEEeqnarray}{C}
\expectation{L'(X)} = H(P) + D(P||Q) + \log \frac{1}{\alpha}, \IEEEeqnarraynumspace
\end{IEEEeqnarray}
where $L'(x)$ is the length of the codeword for symbol $x$;
$P$ is the PMF of the source;
$\alpha$ is defined as $\sum_{x} 2^{-L'(x)}$; and
the PMF $Q$ is given by $Q(x) \triangleq \frac{1}{\alpha} 2^{-L'(x)}$.
The expected codeword length is thus determined by three terms:
an entropy term that depends only on the source;
a divergence term that measures how well the code is matched to the source;
and an inefficiency term that depends only on the code.
(For uniquely decodable codes, $\alpha \le 1$ by Kraft's inequality.)

We have the same structure in (\ref{eq_lma_single_explr}):
an entropy term that depends only on the source;
a divergence term that measures how well the code is matched to the source;
and an inefficiency term that depends only on the code ($\numberoflabels' \le \numberoflabels$ must hold by definition).
\end{remark}

\section{Distributed Task Encoding}
\label{sec_distributed_te}

\begin{theorem}
\label{thm_rate_region}
Recalling the definition of an achievable rate pair from the introduction,
rate pairs $(\ratex,\ratey)$ in the interior of the following region are achievable, while those outside the region are not:
\begin{IEEEeqnarray}{rCl}
\ratex &\ge& \limsup_{n \to \infty} \frac{H_{\rhotilde}(X^{n})}{n}, \label{eq_dte_rateregion_a} \\
\ratey &\ge& \limsup_{n \to \infty} \frac{H_{\rhotilde}(Y^{n})}{n}, \label{eq_dte_rateregion_b} \\
\ratex + \ratey &\ge& \limsup_{n \to \infty} \frac{H_{\rhotilde}(X^{n}, Y^{n}) + K_{\rhotilde}(X^{n}; Y^{n})}{n}. \IEEEeqnarraynumspace \label{eq_dte_rateregion_c}
\end{IEEEeqnarray}
If $\{(X_i,Y_i)\}_{i=1}^{\infty}$ are IID $P_{XY}$, the region (\ref{eq_dte_rateregion_a})--(\ref{eq_dte_rateregion_c}) reduces to
\begin{IEEEeqnarray}{rCl}
\ratex &\ge& H_{\rhotilde}(P_X), \label{eq_dte_iid_rateregion_a} \\
\ratey &\ge& H_{\rhotilde}(P_Y), \label{eq_dte_iid_rateregion_b} \\
\ratex + \ratey &\ge& H_{\rhotilde}(P_{XY}) + K_{\rhotilde}(X; Y). \IEEEeqnarraynumspace \label{eq_dte_iid_rateregion_c}
\end{IEEEeqnarray}
\end{theorem}

\begin{proof}
In both the proof of the converse and the direct part, we use the fact that the set on the RHS of (\ref{eq_intro_lxy_def}) is a Cartesian product, so
\begin{IEEEeqnarray}{C}
L(x^{n},y^{n}) = L_X(x^{n}) L_Y(y^{n}) \IEEEeqnarraynumspace \label{eq_dte_lxy_lx_ly}
\end{IEEEeqnarray}
for all $x^{n} \in \set{X}^{n}$ and $y^{n} \in \set{Y}^{n}$, where
\begin{IEEEeqnarray}{rCl}
L_X(x^{n}) &\triangleq& |\{x' \in \set{X}^{n}: f_n(x') = f_n(x^{n})\}|, \IEEEeqnarraynumspace \\
L_Y(y^{n}) &\triangleq& |\{y' \in \set{Y}^{n}: g_n(y') = g_n(y^{n})\}|.
\end{IEEEeqnarray}

We begin with the converse, i.e., with showing that if a rate pair $(\ratex,\ratey)$ is achievable, then (\ref{eq_dte_rateregion_a})--(\ref{eq_dte_rateregion_c}) must be satisfied.
Observe that
\begin{IEEEeqnarray}{rCl}
\expectation{L(X^{n},Y^{n})^{\rho}} &=& \expectation{L_X(X^{n})^{\rho} L_Y(Y^{n})^{\rho}} \label{eq_dte_rx_a} \\
&\ge& \expectation{L_X(X^{n})^{\rho}} \label{eq_dte_rx_b} \\
&=& 2^{\rho [H_{\rhotilde}(X^{n}) + \Delta_{\rhotilde}(P_{X^{n}}||Q) - \log \numberoflabels']} \IEEEeqnarraynumspace \label{eq_dte_rx_c} \\
&\ge& 2^{\rho n [\frac{1}{n} H_{\rhotilde}(X^{n}) - \ratex]}, \label{eq_dte_rx_d}
\end{IEEEeqnarray}
where (\ref{eq_dte_rx_a}) follows from (\ref{eq_dte_lxy_lx_ly});
(\ref{eq_dte_rx_b}) holds because $L_Y(y^{n}) \ge 1$ for all $y^{n} \in \set{Y}^{n}$;
(\ref{eq_dte_rx_c}) follows from Lemma~\ref{lma_single_cost} applied with the function $f_n\colon \set{X}^{n} \to \{1,\ldots,\lfloor 2^{n \ratex} \rfloor\}$ and the PMF $P_{X^{n}}$; and
(\ref{eq_dte_rx_d}) holds because $\Delta_{\rhotilde}$ is nonnegative \cite{SundaresanGuessing} and because $\numberoflabels' \le \lfloor 2^{n \ratex} \rfloor \le 2^{n \ratex}$.
If (\ref{eq_dte_rateregion_a}) is not satisfied, then there exists a $\gamma > 0$ such that
\begin{IEEEeqnarray}{C}
\frac{1}{n} H_{\rhotilde}(X^{n}) - \ratex \ge \gamma \IEEEeqnarraynumspace \label{eq_dte_ge_gamma}
\end{IEEEeqnarray}
holds for infinitely many values of $n$.
In that case, (\ref{eq_dte_rx_d}) implies that $\limsup_{n \to \infty} \expectation{L(X^{n},Y^{n})^{\rho}} = \infty$, which precludes the possibility that $\lim_{n \to \infty} \expectation{L(X^{n},Y^{n})^{\rho}} = 1$.
Thus, (\ref{eq_dte_rateregion_a}) is necessary for the rate pair $(\ratex,\ratey)$ to be achievable.
The necessity of (\ref{eq_dte_rateregion_b}) follows by swapping the role of $X$ and $Y$ in the above proof.
To see that (\ref{eq_dte_rateregion_c}) is necessary, introduce the PMFs
\begin{IEEEeqnarray}{rCl}
Q_{X^{n}}(x^{n}) &\triangleq& \frac{L_X(x^{n})^{-(1+\rho)}}{\sum_{x' \in \set{X}^{n}} L_X(x')^{-(1+\rho)}}, \IEEEeqnarraynumspace \\
Q_{Y^{n}}(y^{n}) &\triangleq& \frac{L_Y(y^{n})^{-(1+\rho)}}{\sum_{y' \in \set{Y}^{n}} L_Y(y')^{-(1+\rho)}},
\end{IEEEeqnarray}
and observe that
\begin{IEEEeqnarray}{rCl}
\IEEEeqnarraymulticol{3}{l}{\expectation{L(X^{n},Y^{n})^{\rho}}} \nonumber \\* \quad
&=& 2^{\rho [H_{\rhotilde}(X^{n}, Y^{n}) + \Delta_{\rhotilde}(P_{X^{n} Y^{n}}||Q) - \log \numberoflabels']} \label{eq_dte_rxry_a} \\
&=& 2^{\rho [H_{\rhotilde}(X^{n}, Y^{n}) + \Delta_{\rhotilde}(P_{X^{n} Y^{n}}||Q_{X^{n}} Q_{Y^{n}}) - \log \numberoflabels']} \IEEEeqnarraynumspace \label{eq_dte_rxry_b} \\
&\ge& 2^{\rho n [\frac{1}{n} H_{\rhotilde}(X^{n}, Y^{n}) + \frac{1}{n} K_{\rhotilde}(X^{n}; Y^{n}) - (\ratex + \ratey)]}, \label{eq_dte_rxry_c}
\end{IEEEeqnarray}
where (\ref{eq_dte_rxry_a}) follows from Lemma~\ref{lma_single_cost} by viewing the distributed task encoder as a function that maps pairs $(x^{n},y^{n})$ to one of $\lfloor 2^{n \ratex} \rfloor \cdot \lfloor 2^{n \ratey} \rfloor$ labels;
(\ref{eq_dte_rxry_b}) holds because plugging (\ref{eq_dte_lxy_lx_ly}) into (\ref{eq_lma_single_qdef}) leads to $Q(x^{n}, y^{n}) = Q_{X^{n}}(x^{n}) Q_{Y^{n}}(y^{n})$ for all $x^{n} \in \set{X}^{n}$ and $y^{n} \in \set{Y}^{n}$; and
(\ref{eq_dte_rxry_c}) holds since the definition (\ref{eq_def_k_alpha}) implies $K_{\rhotilde}(X^{n}; Y^{n}) \le \Delta_{\rhotilde}(P_{X^{n} Y^{n}}||Q_{X^{n}} Q_{Y^{n}})$ and because we have
$\numberoflabels' \le \lfloor 2^{n \ratex} \rfloor \cdot \lfloor 2^{n \ratey} \rfloor \le 2^{n (\ratex + \ratey)}$.
By the same argument as in (\ref{eq_dte_ge_gamma}), (\ref{eq_dte_rxry_c}) implies that (\ref{eq_dte_rateregion_c}) is necessary for the achievability of the rate pair $(\ratex,\ratey)$.

We next turn to the direct part and show that a rate pair $(\ratex,\ratey)$ is achievable whenever (\ref{eq_dte_rateregion_a})--(\ref{eq_dte_rateregion_c}) all hold with strict inequalities.
We first use the methods from \cite[Section~III-B]{BasicTaskEncoding} to obtain task encoders $f_n$ and $g_n$ that are based on auxiliary PMFs $Q_{X^{n}}$ and $Q_{Y^{n}}$, respectively, and we give bounds on the $\rho$-th moment of the list size.
We then show how to choose $Q_{X^{n}}$ and $Q_{Y^{n}}$ to ensure that (\ref{eq_intro_lim_exp_one}) is satisfied.
Throughout the proof of the direct part, we assume
\begin{IEEEeqnarray}{rCl}
\lfloor 2^{n \ratex} \rfloor - n \log |\set{X}| - 2 &>& 0, \label{eq_dte_n_big_enough_a} \\
\lfloor 2^{n \ratey} \rfloor - n \log |\set{Y}| - 2 &>& 0. \IEEEeqnarraynumspace \label{eq_dte_n_big_enough_b}
\end{IEEEeqnarray}
This entails no loss of generality since we are only interested in the large-$n$ asymptotic performance of our scheme, and because $\ratex$ and $\ratey$ are positive, there exists some $n_0$ such that (\ref{eq_dte_n_big_enough_a}) and (\ref{eq_dte_n_big_enough_b}) hold for all $n \ge n_0$.

Using \cite[Proposition~III.2]{BasicTaskEncoding} twice, we obtain task encoders $f_n\colon \set{X}^{n} \to \{1,\ldots,\lfloor 2^{n \ratex} \rfloor\}$ and $g_n\colon \set{Y}^{n} \to \{1,\ldots,\lfloor 2^{n \ratey} \rfloor\}$
satisfying
\begin{IEEEeqnarray}{rCl}
L_X(x^{n}) &\le& \lambda_X(x^{n}), \IEEEeqnarraynumspace \label{eq_dte_lx_le_lambdax} \\
L_Y(y^{n}) &\le& \lambda_Y(y^{n}) \label{eq_dte_ly_le_lambday}
\end{IEEEeqnarray}
for all $x^{n} \in \set{X}^{n}$ and $y^{n} \in \set{Y}^{n}$, where
\begin{IEEEeqnarray}{rCl}
\beta_X &\triangleq& \frac{2 \sum_{x^{n} \in \set{X}^{n}} Q_{X^{n}}(x^{n})^{\rhotilde}}{\lfloor 2^{n \ratex} \rfloor - n \log |\set{X}| - 2}, \label{eq_dte_betax_def} \\
\beta_Y &\triangleq& \frac{2 \sum_{y^{n} \in \set{Y}^{n}} Q_{Y^{n}}(y^{n})^{\rhotilde}}{\lfloor 2^{n \ratey} \rfloor - n \log |\set{Y}| - 2}, \label{eq_dte_betay_def} \\[.2ex]
\lambda_X(x^{n}) &\triangleq& \begin{cases} \bigl\lceil \beta_X \betaqspace Q_{X^{n}}(x^{n})^{-\rhotilde} \hspace{.1em}\bigr\rceil & \text{if } Q_{X^{n}}(x^{n}) > 0, \label{eq_dte_lambdax_def} \\
\hfil \infty & \text{if } Q_{X^{n}}(x^{n}) = 0, \end{cases} \IEEEeqnarraynumspace \\[.2ex]
\lambda_Y(y^{n}) &\triangleq& \begin{cases} \bigl\lceil \beta_Y \betaqspace Q_{Y^{n}}(y^{n})^{-\rhotilde} \hspace{.1em}\bigr\rceil & \text{if } Q_{Y^{n}}(y^{n}) > 0, \label{eq_dte_lambday_def} \\
\hfil \infty & \text{if } Q_{Y^{n}}(y^{n}) = 0, \end{cases}
\end{IEEEeqnarray}
which are well-defined thanks to (\ref{eq_dte_n_big_enough_a}) and (\ref{eq_dte_n_big_enough_b}).
The condition of \cite[Proposition~III.2]{BasicTaskEncoding} on $\lambda_X$ is fulfilled with $\alpha = 2$ because
\begin{IEEEeqnarray}{rCl}
\IEEEeqnarraymulticol{3}{l}{2 \sum_{x^{n} \in \set{X}^{n}} \frac{1}{\lambda_X(x^{n})} + \log |\set{X}^{n}| + 2} \nonumber \\*[-.4ex] \quad
&\le& 2 \sum_{x^{n} \in \set{X}^{n}} \frac{Q_{X^{n}}(x^{n})^{\rhotilde}}{\beta_X} + n \log |\set{X}| + 2 \IEEEeqnarraynumspace \label{eq_dte_mu_a} \\
&=& \lfloor 2^{n \ratex} \rfloor, \label{eq_dte_mu_b}
\end{IEEEeqnarray}
where (\ref{eq_dte_mu_a}) follows from (\ref{eq_dte_lambdax_def}) and (\ref{eq_dte_mu_b}) follows from (\ref{eq_dte_betax_def}).
The same arguments show that the respective condition on $\lambda_Y$ is also fulfilled.
We upperbound the $\rho$-th moment of the list size as follows (we neglect the issue that $\lambda_X(x^{n})$ and $\lambda_Y(y^{n})$ can be infinite, but it is possible to show that (\ref{eq_dte_rho_moment_e}) continues to hold without this simplification):
\begin{IEEEeqnarray}{rCl}
\IEEEeqnarraymulticol{3}{l}{\expectation{L(X^{n},Y^{n})^{\rho}}} \nonumber \\*[.2ex] \quad \!\!\!
&=& \sum_{x^{n}, y^{n}} P(x^{n}, y^{n}) L_X(x^{n})^{\rho} L_Y(y^{n})^{\rho} \label{eq_dte_rho_moment_a} \\
&\le& \sum_{x^{n}, y^{n}} P(x^{n}, y^{n}) \lambda_X(x^{n})^{\rho} \lambda_Y(y^{n})^{\rho} \label{eq_dte_rho_moment_b} \\
&=& \sum_{x^{n}, y^{n}} P(x^{n}, y^{n}) \left\lceil \frac{\beta_X}{Q_{X^{n}}(x^{n})^{\rhotilde}} \right\rceil^{\rho} \left\lceil \frac{\beta_Y}{Q_{Y^{n}}(y^{n})^{\rhotilde}} \right\rceil^{\rho} \label{eq_dte_rho_moment_c} \\
&\le& \sum_{x^{n}, y^{n}} P(x^{n}, y^{n}) \Biggl\{1 + \frac{2^{\rho} \beta_X^{\rho}}{Q_{X^{n}}(x^{n})^{\rho \rhotilde}} + \frac{2^{\rho} \beta_Y^{\rho}}{Q_{Y^{n}}(y^{n})^{\rho \rhotilde}} \quad \nonumber \\*[-.2ex]
&& \hspace{9.7em} +\> \frac{4^{\rho} \left[\beta_X \beta_Y\right]^{\rho}}{\left[Q_{X^{n}}(x^{n}) Q_{Y^{n}}(y^{n})\right]^{\rho \rhotilde}} \Biggr\} \label{eq_dte_rho_moment_d} \\[.4ex]
&=& 1 + 2^{\rho [H_{\rhotilde}(X^{n}) + \Delta_{\rhotilde}(P_{X^{n}}||Q_{X^{n}}) - n \ratex + \delta_1]} \nonumber \\*
&& +\>2^{\rho [H_{\rhotilde}(Y^{n}) + \Delta_{\rhotilde}(P_{Y^{n}}||Q_{Y^{n}}) - n \ratey + \delta_2]} \nonumber \\*
&& +\>2^{\rho [H_{\rhotilde}(X^{n},Y^{n}) + \Delta_{\rhotilde}(P_{X^{n} Y^{n}}||Q_{X^{n}} Q_{Y^{n}}) - n (\ratex + \ratey) + \delta_3]}, \IEEEeqnarraynumspace \label{eq_dte_rho_moment_e}
\end{IEEEeqnarray}
where (\ref{eq_dte_rho_moment_a}) follows from (\ref{eq_dte_lxy_lx_ly});
(\ref{eq_dte_rho_moment_b}) follows from (\ref{eq_dte_lx_le_lambdax}) and (\ref{eq_dte_ly_le_lambday});
(\ref{eq_dte_rho_moment_c}) follows from (\ref{eq_dte_lambdax_def}) and (\ref{eq_dte_lambday_def});
(\ref{eq_dte_rho_moment_d}) follows from the inequality $\lceil \xi \rceil^{\rho} < 1 + 2^{\rho} \xi^{\rho}$ from \cite[(26)]{BasicTaskEncoding}, which holds for all $\rho > 0$ and $\xi \ge 0$; and
(\ref{eq_dte_rho_moment_e}) follows from (\ref{eq_def_delta_rhotilde}), a longer computation, and the definitions
\begin{IEEEeqnarray}{rCl}
\delta_1 &\triangleq& \log \frac{4 \cdot 2^{n \ratex}}{\lfloor 2^{n \ratex} \rfloor - n \log |\set{X}| - 2}, \IEEEeqnarraynumspace \\
\delta_2 &\triangleq& \log \frac{4 \cdot 2^{n \ratey}}{\lfloor 2^{n \ratey} \rfloor - n \log |\set{Y}| - 2}, \\
\delta_3 &\triangleq& \delta_1 + \delta_2.
\end{IEEEeqnarray}
From (\ref{eq_def_delta_to_d}) we know that
\begin{IEEEeqnarray}{rCl}
\Delta_{\rhotilde}(P_{X^{n}}||Q_{X^{n}}) &=& D_{1+\rho}(\hat{P}_{X^{n}}||\hat{Q}_{X^{n}}) \label{eq_dte_delta_d_x} \\
\Delta_{\rhotilde}(P_{Y^{n}}||Q_{Y^{n}}) &=& D_{1+\rho}(\hat{P}_{Y^{n}}||\hat{Q}_{Y^{n}}) \label{eq_dte_delta_d_y} \\
\Delta_{\rhotilde}(P_{X^{n} Y^{n}}||Q_{X^{n}} Q_{Y^{n}}) &=& D_{1+\rho}(\hat{P}_{X^{n} Y^{n}}||\hat{Q}_{X^{n}} \hat{Q}_{Y^{n}}), \IEEEeqnarraynumspace \label{eq_dte_delta_d_xy}
\end{IEEEeqnarray}
where (\ref{eq_dte_delta_d_xy}) holds because the transformation (\ref{eq_def_q_hat}) of a product is the product of the transformations.
Let $Q_{X^{n}}^{*}$ and $Q_{Y^{n}}^{*}$ be PMFs that achieve equality in (\ref{eq_def_k_alpha}), so
\begin{IEEEeqnarray}{C}
\Delta_{\rhotilde}(P_{X^{n} Y^{n}}||Q_{X^{n}}^{*} Q_{Y^{n}}^{*}) = K_{\rhotilde}(X^{n};Y^{n}). \IEEEeqnarraynumspace \label{eq_dte_def_qstar}
\end{IEEEeqnarray}
We now show how to choose $Q_{X^{n}}$ and $Q_{Y^{n}}$.
Even in the IID case, these will typically not be product distributions.
We consider the mixture distributions
\begin{IEEEeqnarray}{rCl}
\hat{Q}_{X^{n}}(x^{n}) &=& \tfrac{1}{2} \hat{P}_{X^{n}}(x^{n}) + \tfrac{1}{2} \hat{Q}_{X^{n}}^{*}(x^{n}), \IEEEeqnarraynumspace \label{eq_dte_qx_hat} \\
\hat{Q}_{Y^{n}}(y^{n}) &=& \tfrac{1}{2} \hat{P}_{Y^{n}}(y^{n}) + \tfrac{1}{2} \hat{Q}_{Y^{n}}^{*}(y^{n}), \label{eq_dte_qy_hat}
\end{IEEEeqnarray}
and use the inverse transformation (\ref{eq_def_q_hat_inverse}) to obtain $Q_{X^{n}}$ and $Q_{Y^{n}}$.
Consequently,
\begin{IEEEeqnarray}{rCl}
\IEEEeqnarraymulticol{3}{l}{\Delta_{\rhotilde}(P_{X^{n}}||Q_{X^{n}})} \nonumber \\*[.2ex] \quad
&=& D_{1+\rho}(\hat{P}_{X^{n}}||\tfrac{1}{2} \hat{P}_{X^{n}} + \tfrac{1}{2} \hat{Q}_{X^{n}}^{*}) \label{eq_dte_x_div_bound_a} \\
&=& \frac{1}{\rho} \log \sum_{x^{n}} \hat{P}_{X^{n}}(x^{n})^{1+\rho} \left[\frac{\hat{P}_{X^{n}}(x^{n}) + \hat{Q}_{X^{n}}^{*}(x^{n})}{2}\right]^{-\rho} \IEEEeqnarraynumspace \label{eq_dte_x_div_bound_b} \\
&\le& \frac{1}{\rho} \log \sum_{x^{n}} \hat{P}_{X^{n}}(x^{n})^{1+\rho} \left[\frac{1}{2} \hat{P}_{X^{n}}(x^{n})\right]^{-\rho} \label{eq_dte_x_div_bound_c} \\
&=& 1, \label{eq_dte_x_div_bound_d}
\end{IEEEeqnarray}
where (\ref{eq_dte_x_div_bound_a}) follows from (\ref{eq_dte_delta_d_x}) and (\ref{eq_dte_qx_hat});
(\ref{eq_dte_x_div_bound_b}) follows from the definition (\ref{eq_def_renyi_div}); and
(\ref{eq_dte_x_div_bound_d}) holds because $\hat{P}_{X^{n}}$ is a PMF.
In the same way, we obtain $\Delta_{\rhotilde}(P_{Y^{n}}||Q_{Y^{n}}) \le 1$ and
\begin{IEEEeqnarray}{rCl}
\IEEEeqnarraymulticol{3}{l}{\Delta_{\rhotilde}(P_{X^{n} Y^{n}}||Q_{X^{n}} Q_{Y^{n}})} \nonumber \\*[.2ex] \quad
&=& D_{1+\rho}\bigl(\hat{P}_{X^{n} Y^{n}}||(\tfrac{1}{2} \hat{P}_{X^{n}} + \tfrac{1}{2} \hat{Q}_{X^{n}}^{*}) (\tfrac{1}{2} \hat{P}_{Y^{n}} + \tfrac{1}{2} \hat{Q}_{Y^{n}}^{*})\bigr) \IEEEeqnarraynumspace \label{eq_dte_xy_div_bound_a} \\
&\le& \frac{1}{\rho} \log \sum_{x^{n}, y^{n}} \hat{P}(x^{n}, y^{n})^{1+\rho} \left[\frac{\hat{Q}_{X^{n}}^{*}(x^{n}) \hat{Q}_{Y^{n}}^{*}(y^{n})}{4}\right]^{-\rho} \label{eq_dte_xy_div_bound_b} \\
&=& D_{1+\rho}(\hat{P}_{X^{n} Y^{n}}||\hat{Q}_{X^{n}}^{*} \hat{Q}_{Y^{n}}^{*}) + 2 \label{eq_dte_xy_div_bound_c} \\
&=& K_{\rhotilde}(X^{n};Y^{n}) + 2, \label{eq_dte_xy_div_bound_d}
\end{IEEEeqnarray}
where (\ref{eq_dte_xy_div_bound_d}) follows from (\ref{eq_def_delta_to_d}) and (\ref{eq_dte_def_qstar}).
Plugging these results into (\ref{eq_dte_rho_moment_e}), we finally arrive at
\begin{IEEEeqnarray}{rCl}
\IEEEeqnarraymulticol{3}{l}{\expectation{L(X^{n},Y^{n})^{\rho}}} \nonumber \\* \quad \!\!\!
&\le& 1 + 2^{\rho n [\frac{1}{n} H_{\rhotilde}(X^{n}) - \ratex]} \cdot 2^{\rho (\delta_1 + 1)} \nonumber \\*
&& +\>2^{\rho n [\frac{1}{n} H_{\rhotilde}(Y^{n}) - \ratey]} \cdot 2^{\rho (\delta_2 + 1)} \nonumber \\*
&& +\>2^{\rho n [\frac{1}{n} H_{\rhotilde}(X^{n},Y^{n}) + \frac{1}{n} K_{\rhotilde}(X^{n};Y^{n}) - (\ratex + \ratey)]} \cdot 2^{\rho (\delta_3 + 2)}, \IEEEeqnarraynumspace
\end{IEEEeqnarray}
which tends to one as $n$ tends to infinity:
since $(\ratex, \ratey)$ is in the interior of (\ref{eq_dte_rateregion_a})--(\ref{eq_dte_rateregion_c}), the expressions in square brackets will be smaller than or equal to $\gamma$ for some $\gamma < 0$ and $n$ large enough; and
we have $\lim_{n \to \infty} \delta_1 = \lim_{n \to \infty} \delta_2 = 2$ and also $\lim_{n \to \infty} \delta_3 = 4$.

We finish with the specialization of the region (\ref{eq_dte_rateregion_a})--(\ref{eq_dte_rateregion_c}) for an IID source with PMF $P_{XY}$.
In this case, (\ref{eq_dte_iid_rateregion_a})--(\ref{eq_dte_iid_rateregion_c}) readily follow from (\ref{eq_dte_rateregion_a})--(\ref{eq_dte_rateregion_c}) because
\begin{IEEEeqnarray}{rCl}
H_{\rhotilde}(X^{n}) &=& n H_{\rhotilde}(P_X), \label{eq_dte_iid_sl_a} \\
H_{\rhotilde}(Y^{n}) &=& n H_{\rhotilde}(P_Y), \label{eq_dte_iid_sl_b} \\
H_{\rhotilde}(X^{n}, Y^{n}) &=& n H_{\rhotilde}(P_{XY}), \label{eq_dte_iid_sl_c} \\
K_{\rhotilde}(X^{n};Y^{n}) &=& n K_{\rhotilde}(X;Y), \IEEEeqnarraynumspace \label{eq_dte_iid_sl_d}
\end{IEEEeqnarray}
where (\ref{eq_dte_iid_sl_a})--(\ref{eq_dte_iid_sl_c}) follow from the definition (\ref{eq_def_renyi_entropy}) and simple computations; and
(\ref{eq_dte_iid_sl_d}) follows from the repeated application of \cite[Theorem~2, Property~3]{TwoMeasuresOfDependence}.
\end{proof}

\section{On the Definition of the List}
\label{sec_discussion}

To appreciate the subtleties in defining the list, let us first consider the single-source case and compare (\ref{eq_intro_listsize}) with the case where the decoder's list is only required to contain tasks whose probability, conditional on the observed label, is positive.
The list size in this case is
\begin{IEEEeqnarray}{C}
L'(x) \triangleq |\{x' \in \set{X}: P(x') > 0 \> \land \> f(x') = f(x)\}|. \IEEEeqnarraynumspace \label{eq_posterior_listsize}
\end{IEEEeqnarray}

In the single-source case, the two criteria lead to identical asymptotics because for every task encoder $f$ whose $\rho$-th moment of the list size according to (\ref{eq_posterior_listsize}) is $\expectation{L_f'(X)^{\rho}}$,
there exists a task encoder $g$ that has the same $\rho$-th moment of the list size according to (\ref{eq_intro_listsize}) if $g$ is allowed to use one additional label (which is negligible in an asymptotic setting).
Indeed, if
\begin{IEEEeqnarray}{C}
g(x) = \begin{cases} \hfil f(x) & \text{if } P(x) > 0, \\
\numberoflabels + 1 & \text{if } P(x) = 0, \end{cases} \IEEEeqnarraynumspace \label{eq_posterior_def_g}
\end{IEEEeqnarray}
where $\numberoflabels + 1$ denotes the additional label, then
\begin{IEEEeqnarray}{rCl}
\expectation{L_g(X)^{\rho}} &=& \sum_{x: P(x) > 0} P(x) L_g(x)^{\rho} \label{eq_posterior_gef_a} \\
&=& \sum_{x: P(x) > 0} P(x) L_f'(x)^{\rho} \IEEEeqnarraynumspace \label{eq_posterior_gef_b} \\
&=& \expectation{L_f'(X)^{\rho}}, \label{eq_posterior_gef_c}
\end{IEEEeqnarray}
where (\ref{eq_posterior_gef_b}) follows from (\ref{eq_posterior_def_g}) since tasks with $P(x) > 0$ do not share their labels with zero-probability tasks, so $L_g(x)$ is equal to $L_f'(x)$ for all $x \in \set{X}$ with $P(x) > 0$.

In the distributed case, the picture can change dramatically.
To see why, consider an IID source with $X = Y$:
under the positive posterior probability criterion, the decoder's list will only contain pairs that satisfy $x^{n} = y^{n}$, and a careful analysis shows that rate pairs $(\ratex, \ratey)$ satisfying
\begin{IEEEeqnarray}{C}
\ratex + \ratey \ge H_{\rhotilde}(P_{XY}) \IEEEeqnarraynumspace \label{eq_posterior_alternative_region}
\end{IEEEeqnarray}
with strict inequality are achievable, while those not satisfying (\ref{eq_posterior_alternative_region}) are not.
Unless $X$ and $Y$ are deterministic, this region is strictly larger than the region defined by (\ref{eq_dte_iid_rateregion_a})--(\ref{eq_dte_iid_rateregion_c}):
there are no individual constraints on $\ratex$ and $\ratey$, and the constraint on the sum rate does not include the penalty term $K_{\rhotilde}$.

The definition based on (\ref{eq_intro_listsize}) seems easier to analyze and, unless zero-probability tasks are present, the two criteria are equivalent.

\section{Distributed Guessing}
\label{sec_dist_guessing}

As in distributed task encoding, a source generates a sequence of pairs $\{(X_i,Y_i)\}_{i=1}^{n}$ over the finite alphabet $\set{X} \times \set{Y}$.
Using the functions $f_n$ and $g_n$, the sequence $X^{n}$ is described by one of $\lfloor 2^{n \ratex} \rfloor$ labels and the sequence $Y^{n}$ by one of $\lfloor 2^{n \ratey} \rfloor$ labels.
Given a pair of labels, the decoder repeatedly guesses $(x^{n},y^{n})$ until correct.
We are interested in the number of guesses that the decoder needs.
For a fixed $\rho > 0$, a rate pair $(\ratex,\ratey)$ is called achievable if there exists a sequence of encoders $\{(f_n,g_n)\}_{n=1}^{\infty}$ such that the $\rho$-th moment of the number of guesses tends to one as $n$ tends to infinity, i.e., if
$\lim_{n \to \infty} \expectation{G(X^{n},Y^{n})^{\rho}} = 1$.

In the IID case, we show in \cite{DistributedGuessingAndTaskEncoding} that rate pairs $(\ratex,\ratey)$ in the interior of the following region are achievable, while those outside the region are not:
\begin{IEEEeqnarray}{rCl}
\ratex &\ge& H_{\rhotilde}(X|Y), \\
\ratey &\ge& H_{\rhotilde}(Y|X), \\
\ratex + \ratey &\ge& H_{\rhotilde}(X,Y), \IEEEeqnarraynumspace
\end{IEEEeqnarray}
where $H_{\rhotilde}(X|Y)$ is the conditional R\'enyi entropy from \cite{Arimoto}.

\end{document}